\documentclass[12pt,a4paper]{article}
\usepackage[top=2.5cm,bottom=2.5cm,left=2.5cm,right=2.5cm]{geometry}

\usepackage{amssymb,amsmath,graphicx,color,amsthm}
\usepackage{units}
\usepackage{colortbl}
\usepackage[dvipsnames]{xcolor}
\usepackage{cleveref}
\usepackage{enumerate}
\usepackage{subfig}

\newtheorem{theorem}{Theorem}
\newtheorem{lemma}[theorem]{Lemma}

\def\paren#1{\left( #1 \right)}
\def\acc#1{\left\{ #1 \right\}}
\def\floor#1{\left\lfloor #1 \right\rfloor}
\def\ceil#1{\left\lceil #1 \right\rceil}

\renewcommand{\le}{\leqslant}
\renewcommand{\ge}{\geqslant}
\newcommand{\cc}{\cellcolor{Gray!30}}

\newcommand{\RT}{\hbox{\rm RT}}
\newcommand{\CP}{\hbox{$\mathcal{CP}$}}
\newcommand{\A}{\hbox{$\mathbb{A}$}}
\newcommand{\T}{\hbox{$\mathcal{T}$}}
\newcommand{\set}[1]{\ensuremath{\left\{#1 \right\}}}

\begin{document}

\title{On repetition thresholds of caterpillars and trees of bounded degree}

\author
{
  Borut Lu\v{z}ar\thanks{Faculty of Information Studies, Novo mesto, Slovenia. 
    E-Mail: \texttt{borut.luzar@gmail.com}},\quad
  Pascal Ochem\thanks{LIRMM, Universit\'e de Montpellier, Montpellier, France.
    E-Mail: \texttt{pascal.ochem@lirmm.fr}}, \quad    
  Alexandre Pinlou\thanks{LIRMM, Universit\'e Paul-Valery Montpellier 3, Montpellier, France. \newline
    E-Mail: \texttt{alexandre.pinlou@lirmm.fr}},\quad
}

\maketitle

{
  \begin{abstract}
    \noindent The \emph{repetition threshold} is the smallest real
    number $\alpha$ such that there exists an infinite word over a
    $k$-letter alphabet that avoids repetition of exponent strictly
    greater than $\alpha$.  This notion can be generalized to graph
    classes. In this paper, we completely determine the repetition
    thresholds for caterpillars and caterpillars of maximum degree
    $3$. Additionally, we present bounds for the repetition thresholds
    of trees with bounded maximum degrees.
  \end{abstract}
}

\bigskip
{\noindent\small \textbf{Keywords:} Infinite word; Repetition threshold; Graph coloring}

%%%%%%%%%%%%%%%%%%%%%%%%%%%%%%%%%%%%%%%%%%%%%%%%%%%%%%%%%%%%%%%%%%%%%%%%%%% 
%%%%%%%%%%%%%%%%%%%%%%%%%%%%%%%%%%%%%%%%%%%%%%%%%%%%%%%%%%%%%%%%%%%%%%%%%%% 
\section{Introduction}

A word $w$ of \emph{length} $|w|$=$r$ over an alphabet $\mathbb{A}$ is a
sequence $w_1\dots w_r$ of $r$ letters, i.e. $r$ elements of $\mathbb{A}$.
A \textit{prefix} of a word $w = w_1\dots w_r$
is a word $p = w_1\dots w_s$, for some $s \le r$.

A \emph{repetition} in a word $w$ is a pair of words $p$ (called the
\emph{period}) and $e$ (called the \emph{excess}) such that $pe$ is a
factor of $w$, $p$ is non-empty, and $e$ is a prefix of $pe$.
The \emph{exponent} of a repetition $pe$ is $\exp(pe) = \tfrac{|pe|}{|p|}$. A
\emph{$\beta$-repetition} is a repetition of exponent $\beta$. A word is
\emph{$\alpha^+$-free} (resp. \textit{$\alpha$-free}) if it 
contains no $\beta$-repetition such that $\beta>\alpha$ (resp. $\beta\ge\alpha$).

\bigskip

Given $k \ge 2$, Dejean~\cite{Dej72} defined the repetition threshold
$\RT(k)$ for $k$ letters as the smallest $\alpha$ such that there
exists an infinite $\alpha^+$-free word over a $k$-letter alphabet. 

Dejean initiated the study of $\RT(k)$ in 1972 for $k=2$ and $k=3$.
Her work was followed by a series of papers which determine the exact
value of $\RT(k)$ for any $k\ge 2$.

\begin{theorem}[\cite{Car07,CurRam09b,CurRam09,CurRam11,Dej72,MohCur07,Oll92,Pan84,Rao11}]
  \label{thm:path}
  ~
  \begin{enumerate}[$(i)$]
  \item $\RT(2) = 2$~\cite{Dej72}; \label{thm:path_2}
  \item $\RT(3) = \tfrac{7}{4}$~\cite{Dej72};
  \item $\RT(4) = \tfrac{7}{5}$~\cite{Pan84};
  \item $\RT(k) = \tfrac{k}{k-1}$, for $k \ge 5$~\cite{Car07,CurRam09b,CurRam09,CurRam11,MohCur07,Oll92,Pan84,Rao11}.
  \end{enumerate}
\end{theorem}

\bigskip

The notions of $\alpha$-free word and $\alpha^+$-free word have been
generalized to graphs.  A graph $G$ is determined by a set
of vertices $V(G)$ and a set of edges $E(G)$.  A mapping
$c \ : \ V(G) \to \set{1,\dots,k}$ is a \textit{$k$-coloring} of $G$.
A sequence of colors on a non-intersecting path in a $k$-colored graph
$G$ is called a \textit{factor}. A coloring is said to be $\alpha^+$-free
(resp. $\alpha$-free) if
every factor is $\alpha^+$-free (resp. $\alpha$-free).

The notion of repetition threshold 
can be generalized to graphs as follows.
Given a graph $G$ and $k$ colors, 
$$
\RT(k,G)=\inf_{k\textrm{-coloring}\ c}\sup\acc{\exp(w)\,|\, w\ \textrm{is~a~factor~in}\ c}\,.
$$
When considering the repetition threshold over a whole class of graphs $\mathcal{G}$, it is defined as
$$
\RT(k,\mathcal{G})=\sup_{G\in\mathcal{G}}\RT(k,G)\,.
$$
In the remainder of this paper, $\mathcal{P}$, $\mathcal{C}$,
$\mathcal{S}$, $\mathcal{T}$, $\mathcal{T}_k$, $\mathcal{CP}$, and
$\mathcal{CP}_k$ respectively denote the classes of paths, cycles,
subdivisions\footnote{A subdivision of a graph $G$
  is a graph obtained from $G$ by a sequence of edge subdivisions. Note
  that by a graph subdivision, a ``large enough''
  subdivision is always meant.}, trees, trees of maximum degree $k$, caterpillars and
caterpillars of maximum degree $k$.

Since $\alpha^+$-free words are closed under reversal, the repetition thresholds for paths
are clearly defined as $\RT(k,\mathcal{P}) = \RT(k)$, and thus Theorem~\ref{thm:path}
completely determines $\RT(k,\mathcal{P})$.

In 2004, Aberkane and Currie~\cite{AbeCur04} initiated the study of
the repetition threshold of cycles for $2$ letters. Another result of
Currie~\cite{Cur02} on ternary circular square-free word allows to
determine the repetition threshold of cycles for $3$ letters. In 2012,
Gorbunova~\cite{Gor12} determined the repetition threshold of cycles for $k\ge 6$ letters.

\begin{theorem}[\cite{AbeCur04,Cur02, Gor12}]
  \label{thm:cycle}
  ~
  \begin{itemize}
  \item[$(i)$] $\RT(2,\mathcal{C}) = \tfrac{5}{2}$~\cite{AbeCur04};
  \item[$(ii)$] $\RT(3,\mathcal{C}) = 2$~\cite{Cur02};
  \item [$(iii)$] $\RT(k,\mathcal{C}) = 1+\tfrac1{\ceil{\nicefrac{k}{2}}}$, for $k\ge 6$~\cite{Gor12}.
  \end{itemize}
\end{theorem}

Gorbunova~\cite{Gor12} also conjectured that $\RT(4,\mathcal{C}) = \tfrac32$ and $\RT(4,\mathcal{C}) = \tfrac 43$.

For the classes of graph subdivisions and trees, the bounds are
completely determined~\cite{OchVas12}.

\begin{theorem}[\cite{OchVas12}]
  \label{thm:sub}
  ~
  \begin{itemize}
  \item[$(i)$] $\RT(2,\mathcal{S}) = \tfrac{7}{3}$;
  \item[$(ii)$] $\RT(3,\mathcal{S}) = \tfrac{7}{4}$;
  \item[$(iii)$] $\RT(k,\mathcal{S}) = \tfrac{3}{2}$, for $k \ge 4$.
  \end{itemize}
\end{theorem}

\newpage

\begin{theorem}[\cite{OchVas12}]
  \label{thm:tree}
  ~
  \begin{itemize}
  \item[$(i)$] $\RT(2,\mathcal{T}) = \tfrac{7}{2}$;
  \item[$(ii)$] $\RT(3,\mathcal{T}) = 3$;
  \item[$(iii)$] $\RT(k,\mathcal{T}) = \tfrac{3}{2}$, for $k \ge 4$.
  \end{itemize}
\end{theorem}

In this paper, we continue the study of repetition thresholds of trees
under additional assumptions. In particular, we completely determine
the repetition thresholds for caterpillars of maximum degree 3
(\Cref{thm:cp2,thm:cp3,thm:cp4,th:cp3_k}) and for caterpillars of
unbounded maximum degree (\Cref{thm:cp2,thm:cp3}) for every alphabet
of size $k\ge 2$. We determine the repetition thresholds for trees of
maximum degree 3 for every alphabet of size $k\in\{4,5\}$
(\Cref{th:tree3_5}). We finally give a lower and an upper bound on the
repetition threshold for trees of maximum degree 3 for every alphabet
of size $k\ge 6$ (\Cref{th:tree}).  We summarize the results in
Table~\ref{tbl:sum} (shaded cells correspond to our results).
\begin{table}[htp]
  \begin{center}
    \begin{tabular}{|l|c|c|c|c|c|}
      \hline
      & $ |\mathbb{A}| = 2 $ & $|\mathbb{A}| = 3$ & $|\mathbb{A}| = 4$ & $|\mathbb{A}| = 5$ & $|\mathbb{A}| = k$, $k \ge 6$ \\
      \hline
      $\mathcal{P}$ & $2$ & $\nicefrac{7}{4}$  &  $\nicefrac{7}{5}$  &  $\nicefrac{5}{4}$  &  $\nicefrac{k}{k-1}$ \\
      \hline
      $\mathcal{C}$ & $\nicefrac{5}{2}$ & $2$ & $?$ & $?$ & $1+\tfrac1{\ceil{\nicefrac{k}{2}}}$  \\
      \hline
      $\mathcal{S}$ & $\nicefrac{7}{3}$ & $\nicefrac{7}{4}$ & $\nicefrac{3}{2}$ & $\nicefrac{3}{2}$ & $\nicefrac{3}{2}$ \\
      \hline   
      $\mathcal{CP}_3$ & \cc $3$ & \cc $2$ & \cc $\nicefrac{3}{2}$ & \cc $\nicefrac{4}{3}$ & \cc $1+\tfrac1{\ceil{\nicefrac{k}{2}}}$ \\
      \hline
      $\mathcal{T}_3$ & $?$ & $?$ & \cc $\nicefrac{3}{2}$ & \cc $\nicefrac{3}{2}$ & \cc $1 + \tfrac1{2\log k} + o\left(\tfrac1{\log k}\right)$\\
      \hline
      $\mathcal{CP}$ & \cc $3$ & \cc $2$ &  \cc $\nicefrac{3}{2}$ & \cc $\nicefrac{3}{2}$ & \cc $\nicefrac{3}{2}$ \\
      \hline
      $\mathcal{T}$ & $\nicefrac{7}{2}$ & $3$ &  $\nicefrac{3}{2}$ &  $\nicefrac{3}{2}$ &  $\nicefrac{3}{2}$ \\
      \hline
    \end{tabular} 
  \end{center}
  \caption{Summary of repetition thresholds for different classes of graphs. }
  \label{tbl:sum}
\end{table}

\section{Caterpillars}

A \textit{caterpillar} is a tree such that the graph induced by the
vertices of degree at least $2$ is a path, which is called \textit{backbone}.

\begin{theorem}
  \label{thm:cp2}
  $
  \RT(2,\CP)=\RT(2,\CP_3)=3.
  $
\end{theorem} 

\begin{proof}
  First, we show that the repetition threshold is at least $3$.  Note
  that it suffices to prove it for the class of caterpillars with
  maximum degree $3$. Suppose, to the contrary, that $\RT(2,\CP_3)<3$.
  Then, the factor $xxx$ is forbidden for any $x\in\A$.
  Therefore, in any $3$-free $2$-coloring, every vertex colored with
  $x$ has at most one neighbor colored with $x$.
  It follows that four consecutive backbone vertices of degree $3$ 
  cannot be colored $xyxy$ for any $x,y\in\A$,
  since the $3$-repetition $yxyxyx$ appears. 
  The factor $xyx$ is also forbidden.
  Indeed, $xyx$ must extend to $xxyxx$ on the backbone
  since $xyxy$ is forbidden. Then, $xxyxx$ must extend to $yxxyxxy$ on
  the backbone since $xxx$ is forbidden. 
  Finally, $yxxyxxy$ must extend to the $3$-repetition $xyxxyxxyx$ in the
  caterpillar. Thus, the binary word on the backbone must avoid $xxx$
  and $xyx$. So, this word must be $(0011)^\omega$ which is not
  $3$-free, a contradiction. Hence, $\RT(2,\CP_3)\ge3$.

  Now, consider a $2$-coloring of an arbitrary caterpillar such that the
  backbone induces a $2^+$-free word (which exists by
  \Cref{thm:path}$(i)$) and every pendent vertex gets the
  color distinct from the color of its neighbor. Clearly, this
  $2$-coloring is $3^+$-free, and so
  $\RT(2,\CP)\le3$.
\end{proof}

\begin{theorem} 
  \label{thm:cp3}
  $
  \RT(3,\CP)=\RT(3,\CP_3)=2.
  $
\end{theorem} 
\begin{proof}
  We start by proving $\RT(3,\CP_3)\ge2$. So, suppose, for a
  contradiction, that there is a $2$-free $3$-coloring for any
  caterpillar with maximum degree $3$. In every $2$-free $3$-coloring,
  the factor $xyx$ appears on the backbone, since otherwise the word
  on the backbone would be $(012)^\omega$ which is not $2$-free. Then,
  we have no choice to extend the factor $xyx$ to the right
  (see~\Cref{fig:k3cat}). This induces a $2$-repetition $yxzyxz$.

  \begin{figure}
    \centering 
    \includegraphics[scale=0.7]{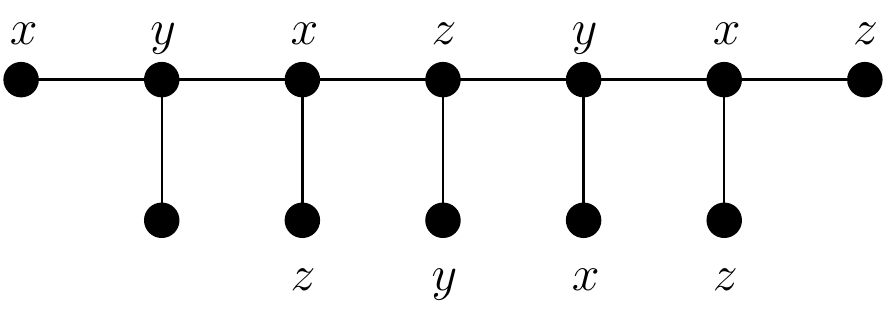}
    \caption{After a factor $xyx$, the remaining colors are forced.}
    \label{fig:k3cat}
  \end{figure}
  
  Now, we show that $\RT(3,\CP)\le2$ by constructing a $2^+$-free
  $3$-coloring of an arbitrary caterpillar. Take a $2^+$-free
  $2$-coloring of the backbone (which exists by Theorem~\ref{thm:path}),
  and color the pendent vertices with the third color.
\end{proof}

\begin{theorem} 
  \label{thm:cp4}
  $
  \RT(4,\CP_3)=\tfrac32.
  $
\end{theorem}

\begin{proof}
  By \Cref{thm:tree}$(iii)$, we have $\RT(4,\CP_3)\le\tfrac32$.  
  Let us show that any $4$-coloring $c$ of a caterpillar of maximum degree
  3 contains a $\tfrac32$-repetition. Consider six consecutive
  vertices $u_0,u_1,u_2,u_3,u_4,u_5$ of the backbone. Let $v_i$ be the pendent
  neighbors of $u_i$. In any $\tfrac32$-free coloring, the vertices
  $u_1,u_2,u_3,v_2$ must get distinct colors: say $c(u_1)=x, c(u_2)=y,
  c(u_3)=z, c(v_2)=t$. Either $u_0$ or $u_4$ must be colored with
  color $t$; w.l.o.g. assume $c(u_4)=t$. Then, either $u_5$ or $v_4$
  must be colored by $y$, and we obtain the $\tfrac53$-repetition $tyzty$.
\end{proof}

\begin{lemma}
  \label{lem:cp35low}
  For every integer $k\ge5$, we have
  $\RT(k,\CP_3)\ge1+\tfrac1{\ceil{\nicefrac k2}}$.
\end{lemma}

\begin{proof}

  \begin{figure}[htp!]
    \centering
    \subfloat[\label{fig:k5catodd}There exist $2\eta$
    vertices at distance at most~$\eta$ from each other.]{\includegraphics[scale=0.7]{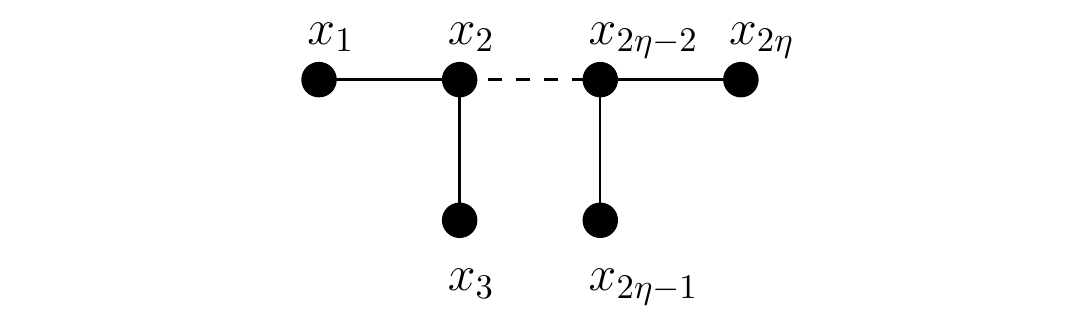}}
    $\qquad\qquad$\subfloat[Even case.\label{fig:k5cateven}]{\includegraphics[scale=0.7]{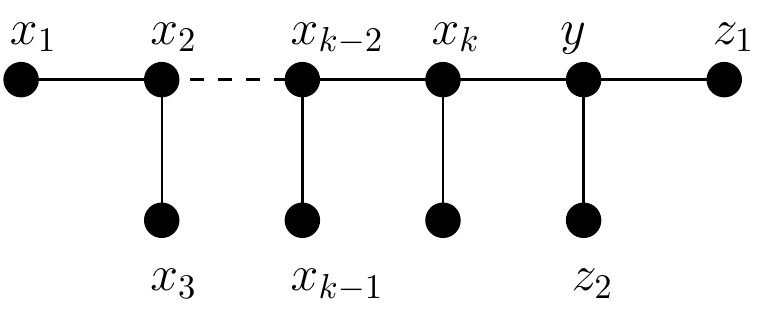}}
    \caption{Illustrations of~\Cref{lem:cp35low}.}
  \end{figure}

  Let $\eta=\ceil{\tfrac k2}$.  Suppose, to the contrary, that there
  exists a $(1+\tfrac1\eta)$-free $k$-coloring $c$ for any caterpillar
  with maximum degree $3$.  Then, every two vertices at distance at
  most $\eta$ must be colored differently. In caterpillars with
  maximum degree $3$, we can have $2\eta$ vertices being pairwise at
  distance at most $\eta$ (see~\Cref{fig:k5catodd}).  If $k$ is odd,
  then $2\eta>k$, and thus $c$ is not $(1+\tfrac1\eta)$-free. If $k$
  is even, the vertices $x_i$ of~\Cref{fig:k5cateven} necessarily get
  distinct colors, say $x_i$ gets color $i$. Then, we have
  $c(y)\in\{1,3\}$ and w.l.o.g. $c(y)=1$. We also have
  $2\in\{c(z_1),c(z_2)\}$ and w.l.o.g. $c(z_1)=2$. Then we obtain a
  $\paren{1+\tfrac2{\eta+1}}$-repetition with excess $c(y)c(z_1)=12$,
  a contradiction.
\end{proof}

\begin{lemma} 
  \label{lem:cp35up}
  $
  \RT(5,\CP_3)\le\tfrac43.
  $ 
\end{lemma}

\begin{proof}
  We start from a right infinite $\tfrac54^+$-free word
  $w=w_0w_1\ldots$ on $5$ letters.  We associate to $w$ its Pansiot
  code $p$ such that $p_i=0$ if $w_i=w_{i+4}$ and $p_i=1$ otherwise,
  for every $i\ge0$~\cite{Pan84}.  Let us construct a
  $\tfrac43^+$-free $5$-coloring $c$ of the infinite caterpillar such
  that every vertex on the backbone has exactly one pendant vertex.
  For every $i\ge0$, $c[0][i]$ is the color of the $i$-th backbone
  vertex and $c[1][i]$ is the color of the $i$-th pendant vertex.

  We define below the mapping $h[t][\ell]$ such that $t\in\acc{0,1}$  
  corresponds to the type of transition in the Pansiot code and
  $\ell\in\acc{0,1}$ corresponds to the type of vertex ($\ell=0$ for  
  backbone, $\ell=1$ for leaf):

  \begin{center}
    \begin{tabular}{c}
      $h[0][0]=150251053150352053$\\
      $h[0][1]=033332322221211110$\\
      $h[1][0]=143123021324123103$\\
      $h[1][1]=000044440400004444$\\
    \end{tabular}
  \end{center}
  
  Notice that the length of $h[t][\ell]$ is $18$. Given
  $t\in\acc{0,1}$ and $\ell\in\acc{0,1}$, let us denote
  $h[t][\ell][j]$, for $j\in\acc{0,\ldots,17}$, the $j^{\rm{th}}$
  letter of $h[t][\ell]$ (e.g. $h[0][0][3] = 2$).

  The coloring is defined by $c[\ell][18i+j]=w_{i+h[p_i][\ell][j]}$
  for every $\ell\in\acc{0,1}$, $i\ge0$, and $j\in\acc{0,\ldots,17}$.
  Let us prove that this coloring is $\tfrac43^+$-free.

  We check exhaustively that there exists no forbidden repetition of
  length at most 576 in the caterpillar. Now suppose for contradiction
  that there exists a repetition $r$ of length $n>576$ and exponent
  $\tfrac n d>\tfrac43$ in the caterpillar. This implies that there
  exists a repetition of length $n'\ge n-2$ and period of length $d$ in
  the backbone.  This repetition contains a repetition $r'$ consisting of
  full blocks of length 18 having length at least
  $n'-2\times({18-1})\ge n-36$ and period length $d$. Given $n>576$
  and $\tfrac n d>\tfrac43$, the repetition $r'$ has exponent at least $\tfrac{n-36}{d}>\tfrac54$.

  The repetition $r'$ in the backbone
  implies a repetition of exponent greater than $\tfrac54$ in $w$, which is a
  contradiction.
\end{proof}

\begin{lemma} 
  \label{lem:cp36up}
  For every integer $k\ge6$, we have 
  $
  \RT(k,\CP_3)\le1+\tfrac1{\ceil{\nicefrac k2}}.
  $
\end{lemma}

\begin{proof}
  \begin{figure}[htp!]
    \centering
    \includegraphics[scale=0.7]{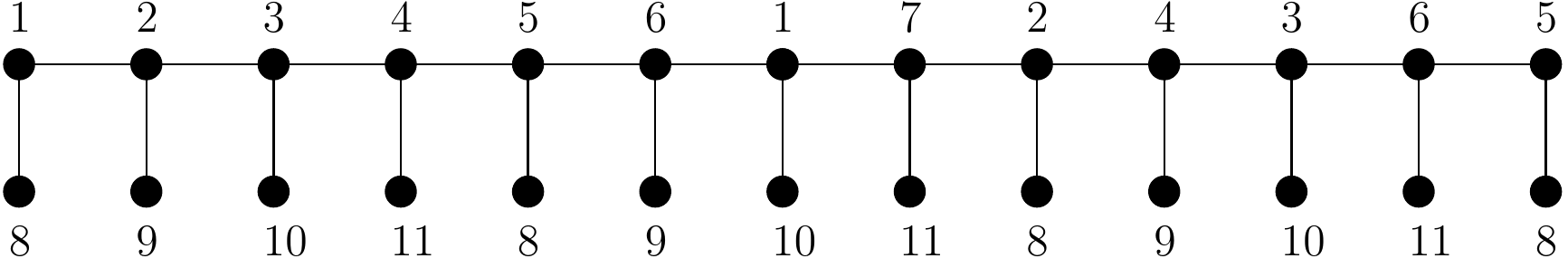}
    \caption{A $(1+\tfrac16)^+$-free $11$-coloring of a caterpillar
      with maximum degree $3$.} 
    \label{fig:k6catup}
  \end{figure}    

  First notice that it suffices to construct colorings for odd $k$'s, since
  $1+\tfrac1{\ceil{\nicefrac k2}}=1+\tfrac1{\ceil{\nicefrac{(k-1)}2}}$
  for $k$ even. So, let $k$ be odd and let $\eta=\ceil{\tfrac k2}$.

  By~\Cref{thm:path}$(iv)$, we can color the vertices of the backbone
  by a $(1+\tfrac1\eta)^+$-free $(\eta+1)$-coloring. Then, it remains
  to color the pendant vertices: let us color them cyclically using
  the remaining $k-(\eta + 1) = \eta-2$ unused colors
  (see~\Cref{fig:k6catup} for an example with $k=11$). Clearly, the
  repetition which does contain a pendant vertex are
  $(1+\tfrac1\eta)^+$-free. Moreover, for a repetition
  containing a pendant vertex, the length of the excess is at most $1$
  and the period length is at least $\eta$. Thus, its exponent is at
  most $\tfrac{\eta+1}{\eta}=1+\tfrac1{\eta}$.

  This shows that this $k$-coloring is $(1+\tfrac1\eta)^+$-free.
\end{proof}

\Cref{lem:cp35low,lem:cp35up,lem:cp36up} together imply the following theorem.

\begin{theorem}\label{th:cp3_k}
  For every integer $k$, with $k\ge5$, we have
  $
  \RT(k,\CP_3)=1+\tfrac1{\ceil{\nicefrac k2}}.
  $
\end{theorem}

Observe that for all $k\ge4$, we have $\RT(k,\CP)=\tfrac32$.
Indeed, caterpillars are trees and thus $\RT(k,\CP)\le\RT(k,\T)=\tfrac32$.
On the other hand, we have $\RT(k,\CP)\ge \RT(k,K_{1,k}) = \tfrac32$ (where $K_{1,k}$ is the star of degree $k$).

\section{Trees of maximum degree 3}

The class of trees of maximum degree 3 is denoted by
$\mathcal{T}_3$. Let $T\in\mathcal{T}_3$ be the infinite embedded rooted tree
whose vertices have degree 3, except the root which has degree 2.
Thus, every vertex of $T$ has a left son and a right son. The \emph{level}
of a vertex of $T$ is the distance to the root (the root has level 0).  

Since every tree of maximum degree 3 is a subgraph of $T$, we only
consider $T$ while proving that $\RT(k,\mathcal{T}_3)\le \alpha$ for
some $k$ and $\alpha$.

\bigskip

Note first that $\RT(4,\mathcal{CP}_3) \le \RT(4,\mathcal{T}_3)  \le \RT(4,\mathcal{T})$ and thus  $\RT(4,\mathcal{T}_3)=\tfrac32$.

\begin{figure}[htp!]
  \centering
  \includegraphics[scale=0.7]{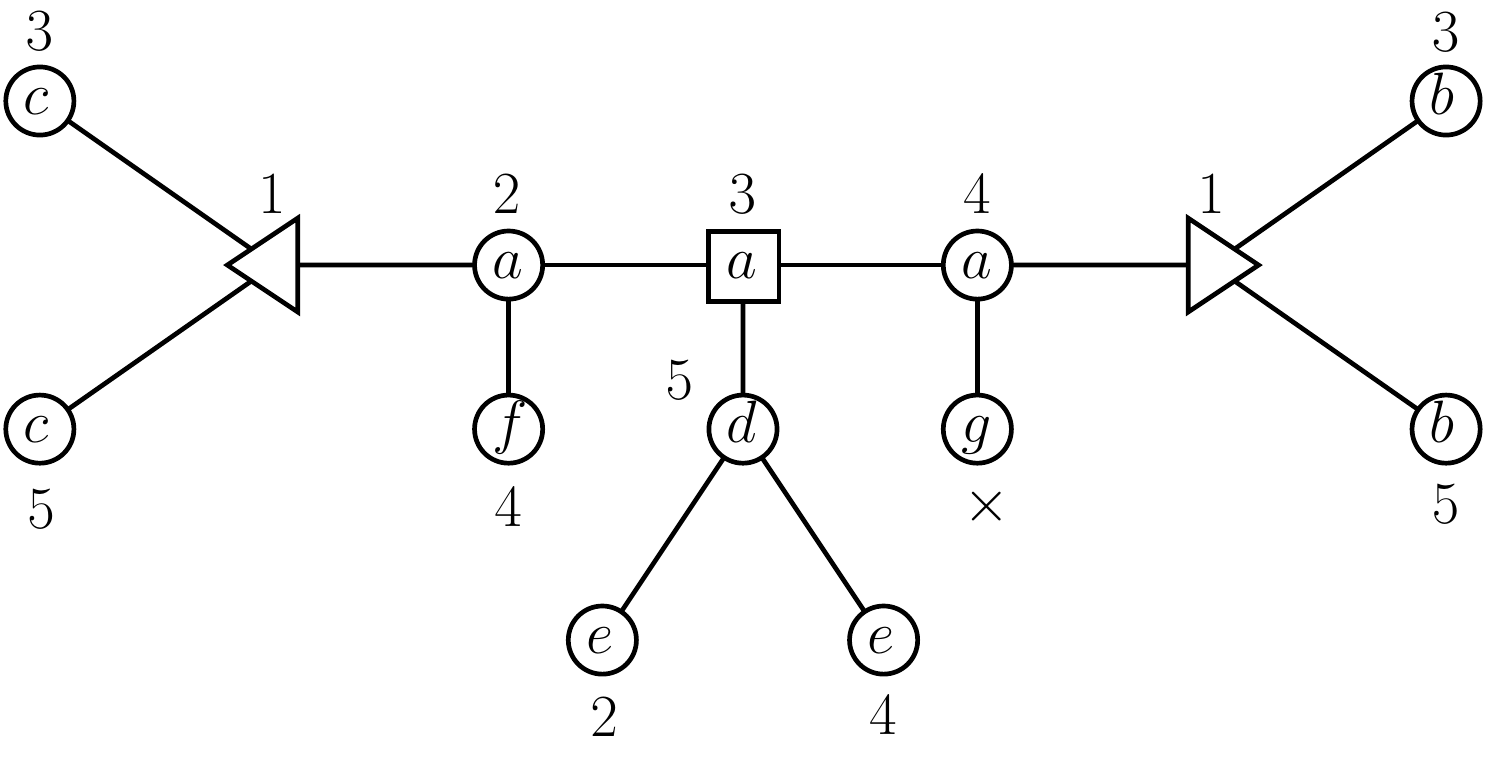}
  \caption{Construction for Theorem~\ref{th:tree3_5}}
  \label{fig:tree} 
\end{figure}

\begin{theorem}\label{th:tree3_5}
  $\RT(5,\mathcal{T}_3) = \tfrac32$.
\end{theorem}

\begin{proof}
  Note first that
  $\RT(5,\mathcal{T}_3) \le \RT(5,\mathcal{T}) = \tfrac32$. Let us
  show that this bound is best possible.
  
  Let $G\in\mathcal{T}_3$ be the graph depicted in
  Figure~\ref{fig:tree} and let $v\in V(G)$ be the squared vertex.
  In every $5$-coloring of $G$, at least two among the six vertices at distance
  $2$ of $v$ will get the same color. In every $\tfrac32$-free
  $5$-coloring, the distance between these two vertices is
  four. W.l.o.g., the two triangle vertices of Figure~\ref{fig:tree}
  are colored with the same color, say color $1$. Then, we color the
  other vertices following the labels in alphabetical order (a vertex
  labelled $x$ is called an $x$-vertex). The $a$-vertices have to get
  three distinct colors (and distinct from~$1$), say $2$, $3$, and
  $4$. The $b$-vertices can only get colors $3$ or $5$ and they must
  have distinct colors in every $\tfrac32$-free $5$-coloring. This is the
  same for $c$-vertices. The $d$-vertex must then get color $5$. Therefore
  the $e$-vertices can only get colors $2$ or $4$. The $f$-vertex must get
  color $4$. Finally, the $g$-vertex cannot be colored without
  creating a forbidden factor. Thus $\RT(5,\mathcal{T}_3)\ge\tfrac32$,
  and that concludes the proof.
\end{proof}

\begin{theorem}\label{th:tree}
  For every $t\ge 4$, we have $$\RT((t+1)2^{\left\lfloor(t+1)/2\right\rfloor},\mathcal{T}_3)\le1+\tfrac1{t}\le\RT\paren{3\paren{2^{\floor{t/2}}-1},\mathcal{T}_3}.$$
\end{theorem}

\begin{proof}
  To prove that
  $\RT((t+1)2^{\left\lfloor(t+1)/2\right\rfloor},\mathcal{T}_3)\le1+\tfrac1{t}$,
  we color $T$ as follows. Let $f$ be the
  coloring of $T$ mapping every vertex $v$ to a color of the form
  $(\gamma,\lambda)$ with $0\le \gamma\le 2t+1$
  and $\lambda\in\{0,1\}^{\floor{(t-1)/2}}$. Let us describe each
  component of a color:
  \begin{description}
  \item[$\gamma$-component:] Let us consider a Dejean word $w$ over
    $t+1$ letters which is $(1+\tfrac 1t)^+$-free since $t\ge 4$. We apply to
    $w$ the morphism $m$ which doubles every letter, that is for every
    letter $i$ such that $0\le i\le t$, $m(i) = i_0i_1$. Let $w'=m(w)$
    and, given $\ell\ge 0$, let $w'[\ell]$ be the $\ell$-th letter of
    $w'$. Every vertex at level $\ell$ gets $w'[\ell]$ as
    $\gamma$-component.

  \item[$\lambda$-component:] Given a vertex $v$, let $u$ be its ancestor at distance
    $\floor{\tfrac{t-1}2}$ (or the root if $v$ is at level
    $\ell<\floor{\tfrac{t-1}2}$). Let
    $u=u_0,u_1,u_2,u_3,\ldots,u_{\floor{\tfrac{t-1}2}}=v$ be the path
    from $u$ to $v$. The $\lambda$-component of $v$ is the binary word built
    as follows: if $u_{i+1}$ is the left son of $u_i$, then $\lambda[i]=0$;
    otherwise, $\lambda[i]=1$. If $v$ is at level $\ell<\floor{\tfrac{t-1}2}$,
    add $\floor{\tfrac{t-1}2} - \ell$ 0's as prefix of $\lambda$.
  \end{description}

  Let us prove that $f$ is a $(1+\tfrac 1t)^+$-free coloring. 

  Suppose that there exists a forbidden repetition such that the
  repeated factor is a single letter $a$, that is a factor $axa$ where
  $|ax|<t$. Let $u$ and $v$ be the two vertices colored
  $f(u)=f(v)=(\gamma,\lambda)$. The vertices $u$ and $v$ must lie on the
  same level, since otherwise they would be at distance at least $2t$
  due to $\gamma$. Since $u$ and $v$ are distinct and have
  the same $\lambda$, their common ancestor is at distance at least
  $\floor{\tfrac{t-1}2}+1$ from each of them. 
  Thus $u$ and $v$ are at distance at least
  $2\paren{\floor{\tfrac{t-1}2}+1}\ge t$, which contradicts $|ax|<t$.

  Suppose now that there exists a forbidden repetition such that the
  length of the repeated factor is at least $2$. Suppose first that the
  path supporting the repetition is of the form
  $l_{i},l_{i-1},l_{i-2},\ldots,l_1,l_0=u=r_0,r_1,r_2,\ldots,r_j$
  where $l_1$ and $r_1$ are the left son and the right son of $u$,
  respectively. Let $l_i,l_{i-1},\ldots,l_1,l_0$ be the 
  \emph{left branch} and $r_0,r_1,\ldots,r_{j-1},r_j$ be the 
  \emph{right branch} of the path.
  W.l.o.g., assume that $1\le j \le i$.
  Therefore, there exist two vertices $l_k$ and $l_{k-1}$ of the left
  branch such that $f(l_kl_{k-1})=f(r_{j-1}r_j)$. Let us show that,
  given the $\gamma$-components of the colors of two adjacent
  vertices, it is possible to determine which vertex is the
  father. W.l.o.g. the two $\gamma$-components are $i_0$ and $j_1$. If
  $i=j$, then the father is the vertex with $\gamma$-component $i_0$;
  otherwise, $i\neq j$ and the father is the vertex with
  $\gamma$-component $j_1$.  This is a contradiction since $l_{k-1}$
  is the father of $l_k$ and $r_{j-1}$ is the father of $r_j$.
  Suppose finally that the path supporting the repetition does not
  contain two brothers. This is equivalent to say that $m(w)$ is not
  $(1+\tfrac 1t)^+$-free. It is clear that if the $m$-image of a word
  contains an $e$-repetition, then this word necessarily
  contains an $e$-repetition. This implies that $w$ is not $(1+\tfrac
  1t)^+$-free, a contradiction. 

  Therefore, we have $\RT((t+1)2^{\left\lfloor(t+1)/2\right\rfloor},\mathcal{T}_3)\le1+\tfrac1{t}$.

  \bigskip

  To prove that $\RT\paren{3\paren{2^{\floor{t/2}}-1},\mathcal{T}_3}\ge1+\tfrac1{t}$,
  we consider the tree $T\in\mathcal{T}_3$ consisting of a vertex and all of its neighbors
  at distance at most $\floor{t/2}$. The distance between every two vertices in $T$ is at most $t$.
  Thus, no two vertices of $T$ have the same color in a $(1+\tfrac1{t})$-free coloring.
  Since $T$ contains $3\paren{2^{\floor{t/2}}-1}+1$ vertices,
  it admits no $(1+\tfrac1{t})$-free coloring with $3\paren{2^{\floor{t/2}}-1}$ colors,
  which gives $\RT\paren{3\paren{2^{\floor{t/2}}-1},\mathcal{T}_3}\ge1+\tfrac1{t}$.

\end{proof}

Note that Theorem~\ref{th:tree} can be generalized to trees of bounded
maximum $\Delta$. This would give the following:
$$\RT(2(t+1)(\Delta-1)^{\left\lfloor(t-1)/2\right\rfloor},\mathcal{T}_\Delta)\le1+\tfrac1{t}\le\RT\paren{\frac{\Delta\paren{(\Delta-1)^{\floor{t/2}}-1}}{\Delta-2},\mathcal{T}_\Delta}$$

\section{Conclusion}

In this paper, we continued the study of repetition thresholds in colorings of various subclasses of trees.
We completely determined the repetition thresholds for caterpillars and caterpillars of maximum degree $3$, 
and presented some results for trees of maximum degree $3$. 
There are several open questions in the latter class for which it appears that more advanced methods of analysis
should be developed. In particular, our bounds show that 
$$
3 \le \RT\paren{2,\mathcal{T}_3} \le \frac{7}{2} \quad \mbox{and} \quad 2 \le \RT\paren{3,\mathcal{T}_3} \le 3\,,
$$
however, we have not been able to determine the exact bounds yet. 
Additionally, the repetition thresholds in trees of bounded degrees 
for alphabets of size at least $6$ remain unknown.

\bigskip
\noindent
{\bf Acknowledgement.} 
The research was partially supported by Slovenian research agency ARRS program no.\ P1--0383 and project no.\ L1--4292
and by French research agency ANR project COCOGRO.

%\bibliographystyle{abbrv}
%\bibliography{MainBib}

\end{document}